\documentclass[12pt,english,draftcls, one column]{IEEEtran}
\usepackage[T1]{fontenc}
\usepackage[latin9]{inputenc}
\usepackage{float}
\usepackage{amsmath}
\usepackage{amsthm}
\usepackage{amssymb}
\usepackage{graphicx}

\makeatletter

\floatstyle{ruled}
\newfloat{algorithm}{tbp}{loa}
\providecommand{\algorithmname}{Algorithm}
\floatname{algorithm}{\protect\algorithmname}

\theoremstyle{plain}

\theoremstyle{plain}

\ifCLASSINFOpdf
\else
\fi
\usepackage{babel}

\makeatother

\usepackage{babel}
\providecommand{\propositionname}{Proposition}
\providecommand{\theoremname}{Theorem}

\begin{document}

\title{Time-Asynchronous Robust Cooperative Transmission for the Downlink of
C-RAN}

\author{Seok-Hwan Park, \textit{Member}, \textit{IEEE}, Osvaldo Simeone,
\textit{Fellow, IEEE}, \\ and Shlomo Shamai (Shitz),\textit{ Fellow,
IEEE}\thanks{S.-H. Park is with the Division of Electronic Engineering, Chonbuk
National University, Jeonju 54896, Korea (email: seokhwan@jbnu.ac.kr).

O. Simeone is with the Center for Wireless Communication and Signal
Processing Research, New Jersey Institute of Technology, 07102 Newark,
New Jersey, USA (email: osvaldo.simeone@njit.edu).

S. Shamai (Shitz) is with the Department of Electrical Engineering,
Technion, Haifa, 32000, Israel (email: sshlomo@ee.technion.ac.il).}
\thanks{The work of S.-H. Park was supported by the National Research Foundation of Korea funded by the Korea Government (MSIP) under Grant NRF-2015R1C1A1A01051825. The work of O. Simeone was partially funded by U.S. NSF through grant CCF-1525629. The work of S. Shamai has been supported by the European Union's Horizon 2020 Research And Innovation Programme, grant agreement no. 694630.
}
\thanks{Copyright (c) 2015 IEEE. Personal use of this material is permitted. However, permission to use this material for any other purposes must be obtained from the IEEE by sending a request to pubs-permissions@ieee.org.
}}
\maketitle
\begin{abstract}
This work studies the robust design of downlink precoding for cloud
radio access network (C-RAN) in the presence of asynchronism among
remote radio heads (RRHs). Specifically, a C-RAN downlink system is
considered in which non-ideal fronthaul links connecting two RRHs
to a Baseband Unit (BBU) may cause a time offset, as well as a phase
offset, between the transmissions of the two RRHs. The offsets are
a priori not known to the BBU. With the aim of counteracting the unknown
time offset, a robust precoding scheme is considered that is based
on the idea of correlating the signal transmitted by one RRH with
a number of delayed versions of the signal transmitted by the other
RRH. For this transmission strategy, the problem of maximizing the
worst-case minimum rate is tackled while satisfying per-RRH transmit
power constraints. Numerical results are reported that verify the
advantages of the proposed robust scheme as compared to conventional
non-robust design criteria as well as non-cooperative transmission.\end{abstract}


\theoremstyle{theorem}
\newtheorem{theorem}{Theorem}
\theoremstyle{proposition}
\newtheorem{proposition}{Proposition}
\theoremstyle{lemma}
\newtheorem{lemma}{Lemma}
\theoremstyle{corollary}
\newtheorem{corollary}{Corollary}
\theoremstyle{definition}
\newtheorem{definition}{Definition}
\theoremstyle{remark}
\newtheorem{remark}{Remark}

\section{Introduction}

Cloud radio access network (C-RAN) is a promising architecture to
address the challenging requirements for the fifth generation (5G)
of wireless communication systems in terms of reduced deployment costs,
high data rate and low power consumption \cite{Simeone-et-al:JCN}.
In a C-RAN, a baseband processing unit (BBU) implements the baseband
processing functionalities of a set of remote radio heads (RRHs) that
are connected to the BBU by means of fronthaul links.

In order to realize these potential benefits of centralized processing
at the BBU, it is necessary to deploy reliable and high-speed fronthaul
links. Nevertheless, cost and technological limitations dictate the
use of fronthauling technologies, such as wireless microwave or mmwave,
that fall short of the ideal requirements of high capacity and perfect
reliability. This has led researchers in both industry and academia
to investigate the impact of fronthaul \textit{capacity} constraints
on the spectral efficiency, see, e.g., \cite{Park-et-al:TSP13}-\cite{Dai-Yu},
as well as the effect of fronthaul \textit{latency} on higher-layer
performance metrics \cite{Khalili-Simeone} (see also review in \cite{Simeone-et-al:JCN}).
This letter contributes to this line of work by studying the implications
of, and countermeasures to, the imperfect mutual \textit{synchronization}
of the RRHs that may result from non-ideal fronthaul connections to
the cloud.

\begin{figure}
\centering\includegraphics[width=12.96cm,height=5.44cm]{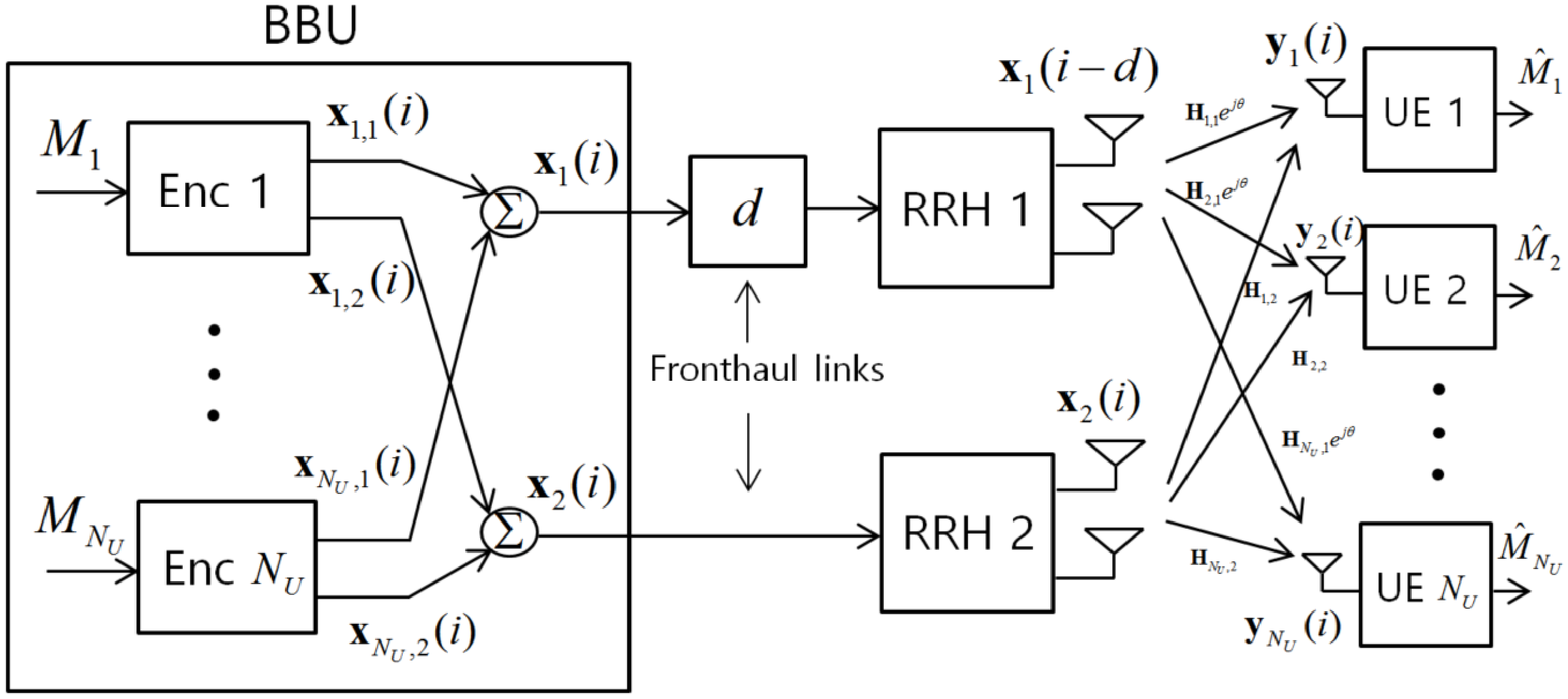}

\caption{\label{fig:system-model}Illustration of a C-RAN downlink with time
offset.}
\end{figure}

To this end, we consider the C-RAN downlink system in Fig. 1 in which
the transmission of the two RRHs is characterized by a relative time
offset, as well as by a phase offset. These offsets are unknown to
the BBU and are generally caused by imperfections in the transmission
and processing of clock-bearing signals on the fronthaul links \cite{Olive-et-al}.
We specifically concentrate on the design of robust cooperative precoding
strategies at the BBU across the two RRHs that account for the existing
uncertainty regarding the inter-RRH time offset. The effect of the
phase offset is also included in the model, but it is not the focus
of the analysis of this letter.

The presence of an unknown time offset generally precludes the use
of cooperative precoding across the two RRHs, since the RRHs only
have radio frequency functionalities and cannot compensate for fronthaul
inaccuracies. Owing to the time offset, the BBU cannot control the
correlation of the signals transmitted by the RRHs, which, in turn,
makes it impossible to ensure the constructive superposition of the
RRHs' signals at the desired receivers. In fact, when the synchronization
is imperfect, it may be beneficial to restrict transmission to non-cooperative
strategies, as studied in \cite{Xu-et-al}.

In this work, we introduce and design a robust asynchronous cooperation
scheme for C-RAN, which is motivated by the coding scheme proposed
in \cite{Yemini-et-al} in the context of an asynchronous cognitive
multiple access channel. In the considered set-up, as detailed in
Sec. II, the time offset is known by the BBU to lie in a bounded range.
The main idea of the robust scheme, as discussed in Sec. III, is to
correlate the signal to be transmitted by one RRH with different delayed
versions of the signal to be transmitted by the other RRH. In this
fashion, no matter what the actual time offset is, partial correlation
between the transmitted signals can be preserved and, with it, cooperative
gains can be potentially accrued. We note that, at a fundamental level,
the idea can be thought of as a robust scheme for communication on
a compound channel \cite{Benammar-et-al}. For this scheme, we tackle
the problem of maximizing the worst-case minimum rate over the correlation
matrices of the transmitted signals while satisfying the per-RRH power
constraints in Sec. IV. Sec. V provides some numerical results that
validate the advantages of the proposed robust scheme as compared
to conventional cooperative and non-cooperative strategies.

\section{System Model\label{sec:System-Model}}

We consider the downlink of a C-RAN shown in Fig. \ref{fig:system-model},
in which a BBU manages two RRHs that communicate with $N_{U}$ user
equipments (UEs). Specifically, by using its fronthaul connections
to the RRHs, the BBU wishes to send a message $M_{k}\in\{1,\ldots,2^{nR_{k}}\}$
to the $k$th UE, where $R_{k}$ and $n$ denote the rate of the message
$M_{k}$ and the coding block length, which is assumed to be sufficiently
large to invoke information-theoretic arguments. We denote the numbers
of the antennas of the $j$th RRH and the $k$th UE as $n_{R,j}$
and $n_{U,k}$, respectively, and define the sets $\mathcal{N}_{U}\triangleq\{1,\ldots,N_{U}\}$
and $\mathcal{N}_{R}\triangleq\{1,2\}$ of the UEs and RRHs, respectively.

Unlike previous works that investigated the impact of fronthaul capacity
or latency limitations (see, e.g., \cite{Simeone-et-al:JCN}), we
study the impact of \textit{asynchronism} among the two RRHs, which
is caused by non-ideal clock transfer or the fronthaul links. We model
the lack of time synchronization with the baseline scenario in Fig.
\ref{fig:system-model}, in which RRH 1 has a time offset of $d$
channel uses with respect to RRH 2, where the amount $d$ is not known
to the BBU. The BBU only knows that the delay $d$ belongs to an uncertainty
set $d\in\mathcal{D}=\{0,1,\ldots,D\}$, where $D\geq1$ represents
the worst-case delay. As in standard C-RAN systems, the RRHs have
only radio frequency functionalities and hence are not able to compensate
for any asynchronism. No further fronthaul limitations, such as in
terms of capacity, are accounted for in the model.

Under a flat fading channel model, the signal $\mathbf{y}_{k}(i)\in\mathbb{C}^{n_{U,k}\times1}$
received by the $k$th UE on the $i$th channel use, $i\in\{1,\ldots,n\}$,
is given as
\begin{equation}
\mathbf{y}_{k}(i)=\mathbf{H}_{k,1}e^{j\theta}\mathbf{x}_{1}(i-d)+\mathbf{H}_{k,2}\mathbf{x}_{2}(i)+\mathbf{z}_{k}(i),\label{eq:received-signal}
\end{equation}
where $\mathbf{x}_{j}(i)\in\mathbb{C}^{n_{R,j}\times1}$ is the signal
transmitted by the $j$th RRH during the $i$th channel use; $\mathbf{H}_{k,j}\in\mathbb{C}^{n_{U,k}\times n_{R,j}}$
is the channel matrix from the $j$th RRH to the $k$th UE; $\theta$
denotes the phase offset between the RRHs, which is also assumed to
be unknown to the BBU; and $\mathbf{z}_{k}(i)\in\mathbb{C}^{n_{U,k}\times1}$
is the additive noise vector distributed as $\mathbf{z}_{k}(i)\sim\mathcal{CN}(\mathbf{0},\mathbf{I})$.
We assume that the channel matrices $\{\mathbf{H}_{k,j}\}_{k\in\mathcal{N}_{U},j\in\mathcal{N}_{R}}$,
the time offset $d$ and the phase offset $\theta$ remain constant
during each coding block of $n$ channel uses, and impose per-RRH
transmit power constraints as $(1/n)\sum_{i=1}^{n}\left\Vert \mathbf{x}_{j}(i)\right\Vert {}^{2}\leq P_{j}$.

As further detailed next, in this work, we focus on the design of
cooperative linear precoding with the aim of ensuring robustness with
respect to the time asynchronicity $d\in\mathcal{D}$. The impact
of the phase offset will be further studied in Sec. \ref{sec:Numerical-Results}
via numerical results.

\section{Asynchronous Robust Transmission\label{sec:Asynchronous-robust-transmission}}

In this section, we describe the proposed asynchronous cooperative
scheme for linear precoding across the two RRHs. The key idea of the
proposed robust scheme is to make the signal $\mathbf{x}_{2}(i)$
transmitted by the RRH 2 correlated to both the current and the delayed
versions $\{\mathbf{x}_{1}(i-d)\}_{d\in\mathcal{D}}$ of the transmit
signal of RRH 1. This follows the idea proposed in \cite{Yemini-et-al},
in which a general information-theoretic formulation was provided
in the context of an asynchronous cognitive channel. The result was
then applied in \cite{Yemini-et-al} to a system with \textit{single-antenna}
transceivers and a \textit{single} user. As mentioned, the robust
scheme is designed to counteract the unknown time offset $d$, but
it assumes no phase offset, i.e., $\theta=0$. Note that, from now
on, we do not explicitly indicate the dependence of the signals on
the channel use index to simplify the notation.

To elaborate, for any channel use, we define a vector $\bar{\mathbf{v}}_{k}=[\mathbf{v}_{k,0}^{\dagger}\,\mathbf{v}_{k,1}^{\dagger}\,\ldots\,\mathbf{v}_{k,D}^{\dagger}]^{\dagger}\in\mathbb{C}^{(D+1)n_{R,1}\times1}$
for each UE $k\in\mathcal{N}_{U}$. This represents $D+1$ consecutive
symbols from the precoded signal transmitted by RRH 1 and intended
for UE $k$, with $\mathbf{v}_{k,i}$ being the signal sent when $d=i$.
Superimposing the signals intended for different UEs, the signal $\mathbf{x}_{1}$
transmitted by RRH 1 when the delay is $d$ is then given as
\begin{equation}
\mathbf{x}_{1}=\sum_{k\in\mathcal{N}_{U}}\mathbf{v}_{k,d}.\label{eq:transmit-signal-RRH-1}
\end{equation}
The signal $\mathbf{x}_{2}$ transmitted by RRH 2 is also given as
the superposition of the signals $\mathbf{x}_{2,k}\in\mathbb{C}^{n_{R,2}\times1}$
intended for each UE $k$ as
\begin{equation}
\mathbf{x}_{2}=\sum_{k\in\mathcal{N}_{U}}\mathbf{x}_{2,k}.\label{eq:transmit-signal-RRH-2}
\end{equation}
Moreover, when the time offset is $d$, the received signal vector
$\mathbf{y}_{k}$ of UE $k$, denoted as $\mathbf{y}_{k,d}$, is given
as
\begin{equation}
\mathbf{y}_{k,d}=\mathbf{H}_{k,1}\sum_{l\in\mathcal{N}_{U}}\mathbf{v}_{l,d}+\mathbf{H}_{k,2}\sum_{l\in\mathcal{N}_{U}}\mathbf{x}_{2,l}+\mathbf{z}_{k}.\label{eq:received-signal-rewritten}
\end{equation}

The key property of the robust precoding scheme is that the signal
$\mathbf{x}_{k,2}$ transmitted by RRH 2 is correlated with the $D+1$
consecutive vectors $\mathbf{v}_{k,d}$, $d\in\mathcal{D}$, each
of which is sent by RRH 1 when the delay is $d$. Specifically, the
vectors $\bar{\mathbf{v}}_{k}$ and $\mathbf{x}_{k,2}$ are characterized
by the covariance matrix
\begin{equation}
\mathbf{\Sigma}_{\mathbf{x}}\!\left(\mathbf{V},\mathbf{\mathbf{\Sigma}}_{\mathbf{x}_{2}},\mathbf{\Omega}\right)\!=\mathbb{E}\!\left[\!\left[\!\!\begin{array}{c}
\bar{\mathbf{v}}_{k}\\
\mathbf{x}_{2,k}
\end{array}\!\!\right]\!\!\left[\!\bar{\mathbf{v}}_{k}^{\dagger}\,\mathbf{x}_{2,k}^{\dagger}\right]\right]\!\!=\!\!\left[\!\!\begin{array}{cc}
\bar{\mathbf{V}}_{k} & \bar{\mathbf{\Omega}}_{k}\\
\bar{\mathbf{\Omega}}_{k}^{\dagger} & \mathbf{\Sigma}_{\mathbf{x}_{2,k}}
\end{array}\!\!\!\right]\!\!,\label{eq:joint-covariance-Gaussian}
\end{equation}
where we defined the correlation matrices $\bar{\mathbf{V}}_{k}=\mathbb{E}[\bar{\mathbf{v}}_{k}\bar{\mathbf{v}}_{k}^{\dagger}]$,
$\mathbf{\Sigma}_{\mathbf{x}_{2,k}}=\mathbb{E}[\mathbf{x}_{2,k}\mathbf{x}_{2,k}^{\dagger}]$
and $\mathbf{\Omega}_{k,d}=\mathbb{E}[\mathbf{v}_{k,d}\mathbf{x}_{2,k}^{\dagger}]$,
and $\bar{\mathbf{\Omega}}_{k}=[\mathbf{\Omega}_{k,0}^{\dagger}\,\mathbf{\Omega}_{k,1}^{\dagger}\,\ldots\,\mathbf{\Omega}_{k,D}^{\dagger}]^{\dagger}$.

Assuming that each UE $k$ decodes its message $M_{k}$ based on the
received signal $\mathbf{y}_{k,d}$ by treating the other signals
as noise, the following proposition derives a vector of rates $\mathbf{R}\triangleq\{R_{k}\}_{k\in\mathcal{N}_{U}}$
that can be supported irrespective of the time offset $d$. The proposition
uses the standard definition for mutual information \cite{ElGamal-Kim}
and is based on assuming the vectors $\bar{\mathbf{v}}_{k}$ and $\mathbf{x}_{2,k}$
to be circularly symmetric complex Gaussian.

\begin{proposition}\label{prop:rate}The following rates are achievable
irrespective of the time offset $d$:
\begin{align}
R_{k} & =\min_{d\in\mathcal{D}}\,\,f_{k,d}\left(\mathbf{V},\mathbf{\mathbf{\Sigma}}_{\mathbf{x}_{2}},\mathbf{\Omega}\right),\label{eq:achievable-rate}
\end{align}
where the function $f_{k,d}(\mathbf{V},\mathbf{\mathbf{\Sigma}}_{\mathbf{x}_{2}},\mathbf{\Omega})$
is defined as
\begin{align}
f_{k,d}\left(\mathbf{V},\mathbf{\mathbf{\Sigma}}_{\mathbf{x}_{2}},\mathbf{\Omega}\right) & =I\left(\mathbf{v}_{k,d};\mathbf{y}_{k,d}\right)+I\left(\mathbf{x}_{2,k};\mathbf{y}_{k,d}|\bar{\mathbf{v}}_{k}\right).\label{eq:rate-function}
\end{align}
The first term $I(\mathbf{v}_{k,d};\mathbf{y}_{k,d})$ in (\ref{eq:rate-function})
can be expressed as
\begin{align}
I\left(\mathbf{v}_{k,d};\mathbf{y}_{k,d}\right) & =\log_{2}\left|\mathbf{V}_{k}\right|+\log_{2}\left|\mathbf{\Sigma}_{\mathbf{y}_{k,d}}\right|-\log_{2}\left|\mathbf{A}_{k,d}\right|,\label{eq:computation-MI-1}
\end{align}
where we defined $\mathbf{V}_{k}=\mathbb{E}[\mathbf{v}_{k,d}\mathbf{v}_{k,d}^{\dagger}]$
and the matrices $\mathbf{\Sigma}_{\mathbf{y}_{k,d}}$ and $\mathbf{A}_{k,d}$
are given as
\begin{align}
\mathbf{\Sigma}_{\mathbf{y}_{k,d}} & =\sum_{l\in\mathcal{N}_{U}}\mathbf{H}_{k,1}\mathbf{V}_{l}\mathbf{H}_{k,1}^{\dagger}+\sum_{l\in\mathcal{N}_{U}}\mathbf{H}_{k,2}\mathbf{\mathbf{\Sigma}}_{\mathbf{x}_{2,l}}\mathbf{H}_{k,2}^{\dagger}\label{eq:covariance-Rx-signal-delay-d}\\
 & +\sum_{l\in\mathcal{N}_{U}}\mathbf{H}_{k,1}\mathbf{\Omega}_{l,d}\mathbf{H}_{k,2}^{\dagger}+\sum_{l\in\mathcal{N}_{U}}\mathbf{H}_{k,2}\mathbf{\Omega}_{l,d}^{\dagger}\mathbf{H}_{k,1}^{\dagger},\nonumber \\
\mathbf{A}_{k,d} & \!=\!\left[\!\!\begin{array}{cc}
\mathbf{V}_{k} & \!\!\!\!\!\!\!\!\!\!\mathbf{V}_{k}\mathbf{H}_{k,1}^{\dagger}+\mathbf{\Omega}_{k,d}\mathbf{H}_{k,2}^{\dagger}\\
\mathbf{H}_{k,1}\mathbf{V}_{k}+\mathbf{H}_{k,2}\mathbf{\Omega}_{k,d}^{\dagger} & \!\!\!\mathbf{\Sigma}_{\mathbf{y}_{k,d}}
\end{array}\!\!\right]\!,\label{eq:covariance-Akd}
\end{align}
and the second term $I(\mathbf{x}_{2,k};\mathbf{y}_{k,d}|\bar{\mathbf{v}}_{k})$
can be written as
\begin{align}
I\left(\mathbf{x}_{2,k};\mathbf{y}_{k,d}|\bar{\mathbf{v}}_{k}\right) & =f_{k,d,2}\left(\mathbf{V},\mathbf{\mathbf{\Sigma}}_{\mathbf{x}_{2}},\mathbf{\Omega}\right)\label{eq:computation-MI-2}\\
 & \!\!\!\!\!\!\!\!\!\!=\log_{2}\left|\mathbf{T}_{k,1}\mathbf{B}_{k,d}\mathbf{T}_{k,1}^{\dagger}\right|-\log\left|\mathbf{B}_{k,d}\right|\nonumber \\
 & \!\!\!\!\!\!\!\!\!\!-(D+1)\log\left|\mathbf{V}_{k}\right|+\log\left|\mathbf{T}_{k,2}\mathbf{B}_{k,d}\mathbf{T}_{k,2}^{\dagger}\right|,\nonumber
\end{align}
where we defined the matrices
\begin{align}
 & \mathbf{B}_{k,d}=\left[\begin{array}{ccc}
\mathbf{\mathbf{\Sigma}}_{\mathbf{x}_{2,k}} & \mathbf{C}_{k,d}^{\dagger} & \bar{\mathbf{\Omega}}_{k}^{\dagger}\\
\mathbf{C}_{k,d} & \mathbf{\Sigma}_{\mathbf{y}_{k,d}} & \mathbf{D}_{k,d}^{\dagger}\\
\bar{\mathbf{\Omega}}_{k} & \mathbf{D}_{k,d} & \bar{\mathbf{V}}_{k}
\end{array}\right],\label{eq:Bkd}\\
 & \mathbf{C}_{k,d}=\mathbf{H}_{k,1}\mathbf{\Omega}_{k,d}+\mathbf{H}_{k,2}\mathbf{\mathbf{\Sigma}}_{\mathbf{x}_{2,k}},\\
 & \mathbf{D}_{k,d}\!=\!\!\left[(\mathbf{H}_{k,2}\mathbf{\Omega}_{k,0}^{\dagger})^{\dagger}\,\cdots\,(\mathbf{H}_{k,2}\mathbf{\Omega}_{k,D}^{\dagger})^{\dagger}\right]^{\dagger}\!\!+\!\!\mathbf{D}_{d}^{\dagger}\mathbf{V}_{k}\mathbf{H}_{k,1}^{\dagger},\\
 & \mathbf{T}_{k,1}=\left[\mathbf{0}_{^{(n_{U,k}+(D+1)n_{R,1})\times n_{R,2}}}\,\,\mathbf{I}_{^{n_{U,k}+(D+1)n_{R,1}}}\right],\label{eq:Tk1}\\
 & \mathbf{T}_{k,2}\!\!=\!\!\!\left[\!\!\!\begin{array}{ccc}
\mathbf{I}_{n_{R,2}} & \!\!\!\!\!\mathbf{0}_{^{n_{R,2}\times n_{U,k}}} & \!\!\!\!\!\!\!\!\mathbf{0}_{^{n_{R,2}\times(D+1)n_{R,1}}}\\
\mathbf{0}_{^{(D+1)n_{R,1}\times n_{R,2}}} & \!\!\!\!\!\mathbf{0}_{^{(D+1)n_{R,1}\times n_{U,k}}} & \!\!\!\!\!\mathbf{I}_{^{(D+1)n_{R,1}}}
\end{array}\!\!\!\!\right]\!\!,\label{eq:Tk2}
\end{align}
with $\mathbf{D}_{d}=[\mathbf{0}_{n_{R,1}\times dn_{T,1}}\,\mathbf{I}_{n_{R,1}}\,\mathbf{0}_{n_{R,1}\times(D-d)n_{R,1}}]$.
We also defined the notations $\mathbf{V}\triangleq\{\mathbf{V}_{k}\}_{k\in\mathcal{N}_{U}}$,
$\mathbf{\mathbf{\Sigma}}_{\mathbf{x}_{2}}\triangleq\{\mathbf{\mathbf{\Sigma}}_{\mathbf{x}_{2,k}}\}_{k\in\mathcal{N}_{U}}$
and $\mathbf{\Omega}=\{\mathbf{\Omega}_{k,d}\}_{k\in\mathcal{N}_{U},d\in\mathcal{D}}$.
\end{proposition}

\begin{proof}It was shown in \cite[Th. 2]{Yemini-et-al} that the
rates $R_{k}$ in (\ref{eq:achievable-rate}) can be achieved with
the function $f_{k,d}(\mathbf{V},\mathbf{\mathbf{\Sigma}}_{\mathbf{x}_{2}},\mathbf{\Omega})$
in (\ref{eq:rate-function}) for a given receiver. The proof is completed
by showing that the mutual information quantities $I(\mathbf{v}_{k,d};\mathbf{y}_{k,d})$
and $I(\mathbf{x}_{k,2};\mathbf{y}_{k,d}|\bar{\mathbf{v}}_{k})$ are
calculated as in (\ref{eq:computation-MI-1}) and (\ref{eq:computation-MI-2}),
respectively.

To this end, we can express $I(\mathbf{v}_{k,d};\mathbf{y}_{k,d})$
and $I(\mathbf{x}_{k,2};\mathbf{y}_{k,d}|\bar{\mathbf{v}}_{k})$ as
\begin{align}
I\left(\mathbf{v}_{k,d};\mathbf{y}_{k,d}\right) & =\!h\!\left(\mathbf{v}_{k,d}\right)+\!h\left(\mathbf{y}_{k,d}\right)-\!h\left(\mathbf{v}_{k,d},\mathbf{y}_{k,d}\right),\label{eq:derivation-MI-1}\\
I(\mathbf{x}_{k,2};\mathbf{y}_{k,d}|\bar{\mathbf{v}}_{k}) & =h\left(\mathbf{x}_{2,k}\right)+h\left(\mathbf{y}_{k,d},\bar{\mathbf{v}}_{k}\right)\!-\!h\left(\mathbf{y}_{k,d},\bar{\mathbf{v}}_{k},\mathbf{x}_{2,k}\right)\nonumber \\
 & -h\left(\mathbf{x}_{2,k}\right)-h\left(\bar{\mathbf{v}}_{k}\right)+h\left(\mathbf{x}_{2,k},\bar{\mathbf{v}}_{k}\right).\label{eq:derivation-MI-2}
\end{align}
Direct calculation of the differential entropy values in (\ref{eq:derivation-MI-1})
and (\ref{eq:derivation-MI-2}) leads to the expressions in (\ref{eq:computation-MI-1})
and (\ref{eq:computation-MI-2}), respectively.\end{proof}

\section{Problem Definition and Optimization\label{sec:Problem-definition-optimization}}

We aim at optimizing the correlation matrices $\mathbf{V}$, $\mathbf{\mathbf{\Sigma}}_{\mathbf{x}_{2}}$
and $\mathbf{\Omega}$ with the goal of maximizing the worst-case
minimum rate $R_{\min}=\min_{k\in\mathcal{N}_{U}}R_{k}$ while satisfying
the per-RRH power constraints. We can state the problem as\begin{subequations}\label{eq:problem-original}
\begin{align}
\underset{\mathbf{V},\mathbf{\mathbf{\Sigma}}_{\mathbf{x}_{2}},\mathbf{\Omega},R_{\min}}{\mathrm{maximize}} & R_{\min}\label{eq:problem-original-objective}\\
\mathrm{s.t.}\,\,\,\,\,\,\, & R_{\min}\leq f_{k,d}\left(\mathbf{V},\mathbf{\mathbf{\Sigma}}_{\mathbf{x}_{2}},\mathbf{\Omega}\right),\,\,\mathrm{for}\,\,k\in\mathcal{N}_{U},\,d\in\mathcal{D},\label{eq:problem-original-rate-constraint}\\
 & \!\!\sum_{k\in\mathcal{N}_{U}}\!\!\mathrm{tr}\!\left(\mathbf{V}_{k}\right)\!\leq\!P_{1},\!\!\,\,\sum_{k\in\mathcal{N}_{U}}\!\!\mathrm{tr}\!\left(\mathbf{\Sigma}_{\mathbf{x}_{2,k}}\right)\!\leq\!P_{2},\label{eq:problem-original-power-constraint}\\
 & \mathbf{\Sigma}_{\mathbf{x}}\left(\mathbf{V},\mathbf{\mathbf{\Sigma}}_{\mathbf{x}_{2}},\mathbf{\Omega}\right)\succeq\mathbf{0},\,\,\mathrm{for}\,\,k\in\mathcal{N}_{U}.\label{eq:problem-original-psd-constraint}
\end{align}
\end{subequations}

\subsection{Optimization\label{sub:Optimization}}

The problem (\ref{eq:problem-original}) is non-convex due to the
constraints (\ref{eq:problem-original-rate-constraint}). However,
it can be seen that the problem is an instance of the difference-of-convex
(DC) problems, and hence we can adopt the concave convex procedure
(CCCP)-based approach as in \cite{Park-et-al:TSP13} to obtain a sequence
of non-decreasing objective values with respect to the number of iterations.
The detailed algorithm is described in Algorithm 1, where we defined
the functions
\begin{align}
 & \tilde{f}_{k,d}\left(\mathbf{V},\mathbf{\mathbf{\Sigma}}_{\mathbf{x}_{2}},\mathbf{\Omega},\mathbf{V}^{(t)},\mathbf{\mathbf{\Sigma}}_{\mathbf{x}_{2}}^{(t)},\mathbf{\Omega}^{(t)}\right)\\
 & =\log_{2}\left|\mathbf{\Sigma}_{\mathbf{y}_{k,d}}\right|+\log_{2}\left|\mathbf{T}_{k,1}\mathbf{B}_{k,d}\mathbf{T}_{k,1}^{\dagger}\right|+\log_{2}\left|\mathbf{T}_{k,2}\mathbf{B}_{k,d}\mathbf{T}_{k,2}^{\dagger}\right|\nonumber \\
 & -D\cdot\Phi\left(\mathbf{V}_{k},\mathbf{V}_{k}^{(t)}\right)-\Phi\left(\mathbf{A}_{k,d},\mathbf{A}_{k,d}^{(t)}\right)-\Phi\left(\mathbf{B}_{k,d},\mathbf{B}_{k,d}^{(t)}\right),\nonumber
\end{align}
with the definition $\Phi(\mathbf{A},\mathbf{B})\!\!=\!\!\log_{2}\left|\mathbf{B}\right|+\mathrm{tr}(\mathbf{B}^{-1}(\mathbf{A}-\mathbf{B}))/\ln2$.

\begin{algorithm}
\caption{CCCP algorithm for problem (\ref{eq:problem-original})}

\textbf{1.} Initialize the matrices $\mathbf{V}^{(1)},\mathbf{\mathbf{\Sigma}}_{\mathbf{x}_{2}}^{(1)},\mathbf{\Omega}^{(1)}$
to arbitrary feasible matrices that satisfy the constraints (\ref{eq:problem-original-power-constraint})-(\ref{eq:problem-original-psd-constraint})
and set $t=1$.

\textbf{2.} Update the matrices $\mathbf{V}^{(t+1)},\mathbf{\mathbf{\Sigma}}_{\mathbf{x}_{2}}^{(t+1)},\mathbf{\Omega}^{(t+1)}$
as a solution of the following convex problem:\begin{subequations}\label{eq:problem-DC}
\begin{align}
\underset{\mathbf{V},\mathbf{\mathbf{\Sigma}}_{\mathbf{x}_{2}},\mathbf{\Omega},R_{\min}}{\mathrm{maximize}} & R_{\min}\label{eq:problem-original-objective-1}\\
\mathrm{s.t.}\,\,\,\,\,\,\, & R_{\min}\leq\tilde{f}_{k,d}\left(\mathbf{V},\mathbf{\mathbf{\Sigma}}_{\mathbf{x}_{2}},\mathbf{\Omega},\mathbf{V}^{(t)},\mathbf{\mathbf{\Sigma}}_{\mathbf{x}_{2}}^{(t)},\mathbf{\Omega}^{(t)}\right),\nonumber \\
 & \,\,\,\,\,\,\mathrm{for}\,\,k\in\mathcal{N}_{U},\,d\in\mathcal{D},\\
 & \!\!\sum_{k\in\mathcal{N}_{U}}\!\!\!\mathrm{tr}\left(\mathbf{V}_{k}\right)\!\leq\!P_{1},\!\!\!\,\,\sum_{k\in\mathcal{N}_{U}}\!\!\!\mathrm{tr}\left(\mathbf{\Sigma}_{\mathbf{x}_{2,k}}\right)\!\leq\!P_{2},\label{eq:problem-original-power-constraint-1}\\
 & \mathbf{\Sigma}_{\mathbf{x}}\left(\mathbf{V},\mathbf{\mathbf{\Sigma}}_{\mathbf{x}_{2}},\mathbf{\Omega}\right)\succeq\mathbf{0},\,\,\mathrm{for}\,\,k\in\mathcal{N}_{U}.\label{eq:problem-original-psd-constraint-1}
\end{align}
\end{subequations}

\textbf{3.} Stop if a convergence criterion is satisfied. Otherwise,
set $t\leftarrow t+1$ and go back to Step 2.
\end{algorithm}

The complexity of Algorithm 1\textbf{ }is given by the product of
the complexity of solving each convex problem (\ref{eq:problem-DC})
and the number of iterations. The complexity of solving a convex problem
is known to be polynomial in the problem size thanks to interior point
algorithms \cite[Ch. 1 and 11]{Boyd}, while the convergence is attained,
in our simulation, within a few tens of iterations. We note that the
analysis of the convergence rate of general CCCP algorithms is still
an open problem to the best of our knowledge.

\subsection{Baseline Schemes\label{sub:Baseline-Schemes}}

In this subsection, we discuss some baseline schemes.

\textit{1) Transmitter Selection:} One could avoid the problem of
the lack of synchronization by activating only the RRH that supports
the largest achievable minimum rate. The performance of this scheme
can be obtained by adopting Algorithm 1 with the additional linear
constraints $\mathbf{V}=\mathbf{0}$ or $\mathbf{\mathbf{\Sigma}}_{\mathbf{x}_{2}}=\mathbf{0}$,
and selecting the solution that yields the best performance.

\textit{2) Non-Cooperative Transmission:} A potentially better approach
that does not require synchronization is to let the RRHs send independent
signals. This approach, referred to as non-cooperative transmission
in Sec. \ref{sec:Numerical-Results}, includes the transmitter selection
scheme as a special case, and the optimization can be addressed by
Algorithm 1 with the additional constraints $\mathbf{\Omega}=\mathbf{0}$
(see also \cite{Xu-et-al} for the optimization).

\textit{3) Asynchronous Conventional Cooperation:} Finally, a conventional
approach would be to design the precoding strategy by assuming that
the time offset $d$ is zero. This design is obtained from Algorithm
1 by setting $\mathbf{\Omega}_{k,d}=\mathbf{0}$ for $k\in\mathcal{N}_{U}$
and $d\geq1$, and removing the constraints (\ref{eq:problem-original-rate-constraint})
with $d\geq1$.

We remark that the complexity increment of the proposed robust scheme
as compared to the baseline strategies depend on the maximum time
offset $D$, since the scheme requires the optimization of the cross-correlation
matrices $\{\mathbf{\Omega}_{k,d}\}_{k\in\mathcal{N}_{U},d\in\mathcal{D}}$,
which are instead not subject to optimization in the baseline strategies.

\section{Numerical Results\label{sec:Numerical-Results}}

\begin{figure}
\centering\includegraphics[width=9cm,height=7cm]{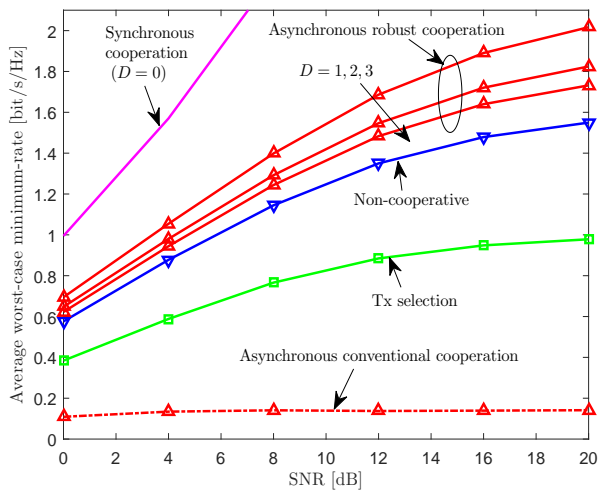}\caption{{\scriptsize{}\label{fig:graph-as-SNR}Average worst-case rate $R_{\min}$
versus the SNR $P$ ($n_{R,i}=n_{U,k}=1$ and $N_{U}=2$).}}
\end{figure}

In this section, we present some exemplifying numerical results to
gauge the potential advantages of the proposed robust scheme as compared
to the conventional solutions discussed in Sec. \ref{sub:Baseline-Schemes}.
To this end, we consider a symmetric system where the RRHs use the
same transmit power $P_{1}=P_{2}=P$ and the elements of the channel
matrices $\mathbf{H}_{k,j}$ are sampled in an independent and identically
distributed (i.i.d.) manner from a complex Gaussian distribution with
zero mean and unit variance. We used the cvx software \cite{Boyd:cvx}
to solve the convex problem (\ref{eq:problem-DC}) at each iteration
of Algorithm 1, and the maximum number of iterations was set to 50.
The worst-case minimum rate was averaged over 100 channel realizations.

Fig. \ref{fig:graph-as-SNR} shows the average worst-case rate versus
the signal-to-noise ratio (SNR) $P$ for the downlink of a C-RAN with
$n_{R,i}=n_{U,k}=1$ and $N_{U}=2$. We set here $\theta=0$, that
is, zero phase offset. For reference, we also show the performance
of the synchronous cooperation scheme for which the delay $d$ is
perfectly known. It is observed that the proposed robust scheme improves
over the non-cooperative and transmitter selection schemes with a
gain increasing with the SNR. This is in line with the intuition that
precoding design becomes more critical in the interference-limited
regime of high SNR. However, the performance gain of the proposed
robust scheme is reduced as the worst-case time offset $D$ increases.
We also emphasize that the conventional non-robust cooperation scheme
that neglects the uncertainty on the time offset performs even worse
than the transmitter selection scheme, highlighting the importance
of the robust design.

\begin{figure}
\centering\includegraphics[width=9cm,height=7cm]{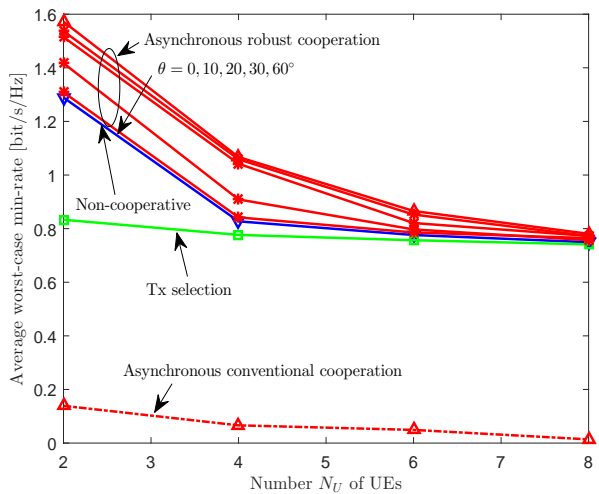}\caption{{\scriptsize{}\label{fig:graph-as-numUE}Average worst-case rate $R_{\min}$
versus the number $N_{U}$ of UEs ($n_{R,i}=N_{U}/2$, $n_{U,k}=1$
and $P=10$ dB).}}
\end{figure}

Fig. \ref{fig:graph-as-numUE} plots the average worst-case rate versus
the number $N_{U}$ of UEs for the downlink of a C-RAN with $n_{R,i}=N_{U}/2$,
$n_{U,k}=1$ and $P=10$ dB. Here we consider various values of the
phase offset $\theta$ as indicated in the figure. We emphasize that
$n_{R,i}=N_{U}/2$ is the smallest number of RRH antennas that are
able to serve the $N_{U}$ UEs without creating inter-UE interference
in the presence of perfect synchronization. Note that the achievable
rates can be easily computed in the presence of a phase offset by
using (\ref{eq:received-signal}) in (\ref{eq:rate-function}) as
done in the proof of Proposition \ref{prop:rate}. We can see that,
although the number $n_{R,i}$ of RRH antennas scales with the number
of UEs, the performance of all the schemes is interference-limited
as long as perfect synchronization is not available. Nevertheless,
for a sufficiently small number of UEs, robust cooperation yields
significant performance gains, even for phase offsets as high as $20\textdegree$.
However, for larger phase offsets, the observed performance degradation
calls for the development of a precoding design that is robust to
the phase offset.

\section{Concluding Remarks\label{sec:Concluding-Remarks}}

Fronthaul limitations constitute one of the key bottlenecks in the
implementation of C-RAN. In this work, we have considered the aspect
of imperfect RRH time synchronization, which may be caused by imperfect
clock distribution through the fronthaul network. The proposed robust
precoding design was demonstrated to have important advantages, particularly
in the high-SNR regime. Among the many interesting open issues, we
mention here the design of precoding strategies that are robust to
both time and phase offsets, as well as the analysis of models with
an arbitrary number of RRHs. 

\end{document}